\def\y{\mathbf{y}}

\def\z{\mathbf{z}}

\def\x{\mathbf{x}}
\def\t{\mathbf{t}}
\def\g{\mathbf{g}}

\def\M{\mathbf{M}}
\def\A{\mathbf{A}}
\def\I{\mathbf{I}}

\def\Q{\mathbf{Q}}

\documentclass[a4paper,10pt,twocolumn]{article}

\usepackage[english]{babel}
\usepackage[utf8]{inputenc}
\usepackage[T1]{fontenc}
\usepackage{algorithm}
\usepackage{algorithmic}

\usepackage{amsthm}
\newtheorem{theorem}{Theorem}[section]

\usepackage[top=1.5cm, left=1.5cm, right=1.5cm, bottom=1.5cm]{geometry}

\renewenvironment{abstract}{\bf\small {\em\ Abstract---}}{}
\usepackage{amsfonts,amssymb,amsmath,amsthm}
\usepackage{subfigure}
\usepackage{graphicx}
\usepackage[footnotesize]{caption}
\usepackage{amsfonts,amssymb,amsmath,amsthm}

\title{BSGD-TV: A parallel algorithm solving total variation constrained image reconstruction problems}

\author{Yushan Gao$^1$, Thomas Blumensath$^2$.\\
  \footnotesize $^1$ University of Southampton UK.\ $^2$University of Southampton UK.\
  } \date{\empty} 

\begin{document}

\maketitle

\begin{abstract} 
We propose a parallel reconstruction algorithm to solve large scale TV constrained linear inverse problems. We provide a convergence proof and show numerically that our method is significantly faster than the main competitor, block ADMM. 
\end{abstract}

\section{Introduction}
\label{sec:introduction}
Our algorithm is inspired by applications in computed tomography (CT), where the efficient inversion of large sparse linear systems is required \cite{guo2016convergence}: 
$\y\approx \A\x_{true}$, where $\x_{true} \in \mathbb{R}^c$ is the vectorised version of a 3D image that is to be reconstructed and $\A\in \mathbb{R}^{r\times c}$ is an X-ray projection model. $\y\in \mathbb R^r$ are the vectorised noisy projections. We are interested in minimizing $f(\x)+g(\x)$, where $f(\x)$ is quadratic and  $g(\x)$ is convex but non-smooth  \cite{sidky2008image}. For example:
\begin{equation}
\x^{\star}=\arg \min_{\x} \underbrace{(\y-\A\x)^T(\y-\A\x)}_{f(\x)}
+\underbrace{2\lambda \text{TV}(\x)}_{g(\x)},
\label{Equ}
\end{equation}
where $\lambda$ is a relaxation parameter and $\text{TV}(\x)$ is the total variation (TV) of the image $\x$. For 2D images, it is defined as: 
\begin{equation}
\text{TV}(\x)=\sum_{s,t}\sqrt{(x_{s,t}-x_{s-1,t})^2+(x_{s,t}-s_{s,t-1})^2},
\end{equation}
where $x_{s,t}$ is the intensity of image pixel in row $s$ and column $t$.

We recently introduced a parallel reconstruction algorithm called coordinate-reduced stochastic gradient descent (CSGD) to minimize  quadratic objective function $f(\x)$ \cite{gao2017distributed}. We here introduce a slight modification by simplifying the step length calculation and show that the  modified version converges to the least squares solution of $f(\x)$. We will call this modified algorithm block stochastic gradient descend (BSGD). We combine BSGD with an iterative shrinkage/thresholding (ISTA-type) step \cite{beck2009fast} to solve Eq.\ref{Equ}. The new algorithm, called BSGD-TV, is compared with block ADMM-TV \cite{parikh2014block}, an algorithm sharing the same parallel architecture and the same communication cost. Simulation results show that BSGD-TV is significantly faster as it requires significantly fewer matrix vector products compared to block ADMM-TV.

\section{BSGD-TV Algorithm}
\label{sec:first-section}
\subsection{Algorithm description}
BSGD works on blocks of $\x$ and $\y$. We assume that $\A$ is divided into $M$ row blocks and $N$ column blocks. Let $\{\x_{J_j}\}_{j=1}^N$ and $\{\y_{I_i}\}_{i=1}^M$ be sub-vectors of $\x$ and $\y$ and let $\A_{I_i}^{J_j}$ be the associated block of matrix $\A$ so that $\y_{I_i} \approx \sum_{j=1}^N \A_{I_i}^{J_j} \x_{J_j}$. Our algorithm splits the optimization into blocks, so that each parallel process only computes using a single block $\x_{J_j}$ and $\y_{I_i}$ for some $J_j\in \{J_j\}_{j=1}^N$ and $I_i\in \{I_i\}_{i=1}^M$. Each process also requires an estimate of the current residual $\mathbf{r}_{I_i}$ and computes a vector $\z^j_{I_i}$, both of which are of the same size as $\y_{I_i}$. The main steps (ignoring initialisation) are described in Algo.\ref{BSGDIJ}.
\begin{algorithm}  
  \caption{BSGD-TV algorithm}
   \label{BSGDIJ}  
    \begin{algorithmic}[1]  
     \FOR{$epoch =1,2,\cdots$}
        \FOR{$M\times N$ pairs $\{i,j\}$ drawn randomly  without replacement \textbf{in parallel} } 
    			\STATE $\hat{\mathbf{g}}_{J_j}^i=2(\A_{I_i}^{J_j})^T\mathbf{r}_{I_i}$
       			\STATE $\z_{I_i}^j=\A_{I_i}^{J_j}{\x}_{J_j}$
   	\ENDFOR
   	\STATE $\mathbf{r} = \y-\sum_j \z^j$
   	\STATE $\g = \sum_i \hat{\mathbf{g}}^i$
   	\STATE $\x=\x+\mu \g$ (constant $\mu>0$)
   	\STATE $\x=\arg \min_{\t} \Vert\t-\x\Vert^2+2\mu\lambda \text{TV}(\t)$
 \ENDFOR
\end{algorithmic}  
\end{algorithm}
To effectively solve line 9, we here adopt method proposed in \cite{beck2009fast2}. 

\subsection{BSGD Convergence}
BSGD without the proximal operator ($\lambda$=0 in Eq.\ref{Equ}), and with parallelization over all subsets can be shown to converge to the least squares solution. To see this, we write the update of $\x$ as
\begin{equation}
\begin{aligned}
\x^{k+1}&= \x^k+\mu\g^k\\
        &=\x^k+ 2\mu \A^T(\y-\sum_{j=1}^N(\z^j)^{k-1})\\
        &=\x^k+2\mu \A^T(\y-\A\x^{k-1})
\end{aligned}
\label{TotalEq}
\end{equation}
In this form, BSGD is similar to gradient descent but uses an old gradient. Assume that there is a fixed point $\x^{\star}$ defined by $\x^{\star}=\x^{\star}+2\mu \A^T(\y-\A\x^{\star})$. Note that the fixed point condition implies that, if $\A$ is full column rank, then $\x^{\star}=(\A^T\A)^{-1}\A^T\y$. Thus the fixed point is the least squares solution. 
Theorem  \ref{thm1} states the conditions on parameter $\mu$ for convergence when all subsets $\{I_i\}_{i=1}^M$ and $\{J_j\}_{j=1}^N$ are selected within one epoch. 
\begin{theorem}
\label{thm1}
If $\mu \in (0,\frac{1}{2u_{max}})$, where $u_{max}$ is the maximum eigenvalue of $\A^T\A$ and assume $\A$ is full column rank, then BSGD without the TV operator ($\lambda$=0 in Eq.\ref{Equ}), and with parallelization over all subsets  converges to the least squares solution $\x^\star$.
\end{theorem}

 

\begin{proof}[Proof of Theorem \ref{thm1}]
 The iteration  in Eq.\ref{TotalEq} can be written as
\begin{equation}
\begin{aligned}
\begin{bmatrix}
\x^{k}\\
\x^{k+1}
\end{bmatrix}=&\begin{bmatrix}
\mathbf{0} \quad \I\\
-2\mu \A^T\A \quad \I
\end{bmatrix}
\begin{bmatrix}
\x^{k-1}\\
\x^{k}
\end{bmatrix}+ \begin{bmatrix}
\mathbf{0}\\
2\mu \A^T\y
\end{bmatrix}\\
&=\M\begin{bmatrix}
\x^{k-1}\\
\x^{k}
\end{bmatrix}+\begin{bmatrix}
\mathbf{0}\\
2\mu \A^T\y
\end{bmatrix}.
\end{aligned}
\end{equation}
Standard convergence results for iterative method of this type with fixed $\M$ require the spectral radius of $\M$ to be less than 1 \cite{saad2003iterative}. 
Let $v$ be any (possibly complex valued) eigenvalue of $\M$, i.e. $v$ satisfies det$(\M-v\I)=0$. It is straightforward to obtain:
\begin{equation}
\text{det}\left( \begin{bmatrix}
-v\I \quad \I\\
-2\mu \A^T\A \quad I-v\I
\end{bmatrix}\right)=\text{det}(\A^T\A-\frac{v-v^2}{2\mu}\I)=0
\label{det1}
\end{equation}
By Eq.\ref{det1}, we see that eigenvalues $u$ of $\A^T\A$  correspond to 

 \begin{equation}
u=\frac{v-v^2}{2\mu},
 \end{equation}
Eigenvalues of $\M$ are then given by 
 \begin{equation}
 \begin{aligned}
 v_1=\frac{1+\sqrt{1-8\mu u}}{2},
 v_2=\frac{1-\sqrt{1-8\mu u}}{2}.
 \end{aligned}
 \end{equation}

As the spectral radius of $\M$ corresponds to the largest magnitude of the eigenvalues of $\M$, we require $\vert v_1\vert<1$ and $\vert v_2\vert<1$ to ensure the convergence of the algorithm.
 $\A^T\A$ is a positive definite matrix and thus has only positive, real valued eigenvalues $u$. 
Thus $v_1$ and $v_2$ are real valued if $0< \mu \leq\frac{1}{8u}$ and complex valued if $\mu$ is $\frac{1}{8u} < \mu$. In the complex case, it is easy to see that $\vert v_1\vert<1$ and $\vert v_2\vert<1$ if $\mu <\frac{1}{2u}$, implying that the acceptable range of $\mu$ is $(0,\frac{1}{2u_{max}})$. 
\end{proof}
 
Theorem \ref{thm2} gives a general convergence condition when applying BSGD-TV to solve Eq.\ref{Equ}.

\begin{theorem}
\label{thm2}
If the constant step length $\mu$ satisfies
\begin{equation}
f(\x^{k+1})<f(\x^k)+(\x^{k+1}-\x^k)^T\nabla f(\x^{k-1})+\frac{1}{2\mu}\|\x^{k+1}-\x^k\|^2,
\label{con}
\end{equation}
where $f(\x)$ is defined in Eq.\ref{Equ}, then BSGD-TV converges to the optimal solution of Eq.\ref{Equ}.
\end{theorem}
\begin{proof}[Proof of Theorem \ref{thm2}]
With parallelization over all subsets, BSGD-TV computes
\begin{equation}
\begin{aligned}
& \hat{\x}^{k+1}=\x^k - \mu \nabla f(\x^{k-1}) \\
& \x^{k+1}=\arg \min_{\x}\{2\mu g(\x) + \|\x-\hat{\x}^{k+1}\|\}.
\end{aligned}
\end{equation}
We define a function $\Q$ as
\begin{equation}
\begin{aligned}
\Q(\x,\x^{k},\x^{k-1})=&f(\x^k)+(\x-\x^k)^T \nabla f(\x^{k-1})\\
&+\frac{1}{2\mu}\|\x-\x^k\|^2+g(\x),
\end{aligned}
\end{equation}
where $\|\cdot \|^2$ is the squared $\ell_2$ norm.
The fact that
$
\arg \min_{\x}\{\Q(\x,\x^k,\x^{k-1})\}\equiv \x^{k+1}
$
means that:
\begin{equation}
\Q(\x^{k+1},\x^k,\x^{k-1})<\Q(\x^k,\x^k,\x^{k-1})\equiv f(\x^k)+g(\x^k)
\end{equation}
Finally, the definition of $\Q$ and the requirement on the step length $\mu$ in Eq.\ref{con},
mean that $f(\x^{k+1})+g(\x^{k+1})<\Q(\x^{k+1},\x^k,\x^{k-1})<f(\x^k)+g(\x^k)$ holds, so that BSGD-TV converges to the fixed point of Eq.\ref{Equ}.
\end{proof}

\section{Simulations}
We show experimentally that the method also converges when only a fraction $\alpha$ and $\gamma$ of subsets of $\{\x_{J_j}\}_{j=1}^N$ and $\{\y_{I_i}\}_{i=1}^M$ are randomly selected to calculate the $\g$ at each iteration.
The simulation geometry is shown in Fig.\ref{geo}.
\begin{figure}[htb]
\centering
\includegraphics[width=6cm,angle=0]{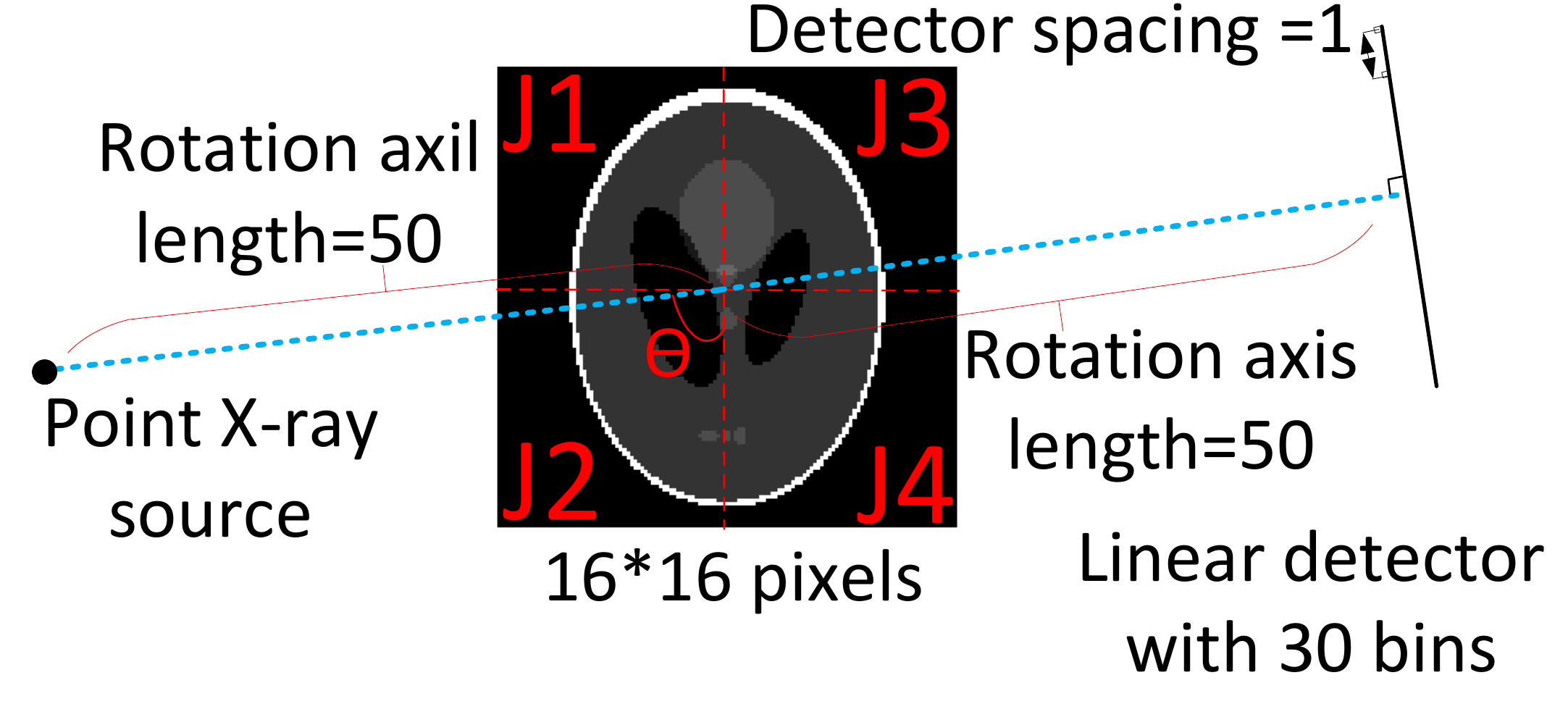}
\caption{Using a fan-beam x-rays geometry to scan a 2D Shepp-Logan phantom. Projections are taken  at $10^\circ$ intervals. The image $\x$ is partitioned into 4 subsets $\{J_j\}_{j=1}^4$ and the total projections are also partitioned into 4 subsets.
}
\label{geo}
\end{figure}
We add Gaussian noise to the projections so that the SNR of $\y$ is 17.7 dB. We define the $\textbf{relative error}$ as $\frac{\|\x_{dif}\|}{\|\x_{true}\|}$, where $\Vert\x_{dif}\Vert$ is the $\ell_2$ norm of the difference between reconstructed image vector and the original vector $\x_{true}$.  
Convergence is shown in Fig.\ref{figsubfig}a. We plot relative error against epochs, where an epoch is a normalised iteration count that corrects for the fact that the stochastic version of our algorithm only updates a subset of elements at each iteration. 

BSGD-TV and ADMM-TV are faster than ISTA in terms of epochs. However, ADMM-TV is significantly slower than BSGD-TV, because ADMM-TV requires matrix inversions at each iteration, while BSGD does not. Even when implementing ADMM-TV using as few conjugate gradient iterations per step as possible, as shown in Fig.\ref{figsubfig}b, BSGD-TV is more computationally efficient in terms of the number of required matrix vector multiplications \cite{gao2017distributed}. Compared to ISTA and GD, our block method allows these computations to be fully parallelised which would enable to reconstruct large scale CT reconstructions while the computation node have limited storage capacity.

\begin{figure}[htb] 
  \centering 
  \subfigure[]{ 
    \includegraphics[width=1.65in]{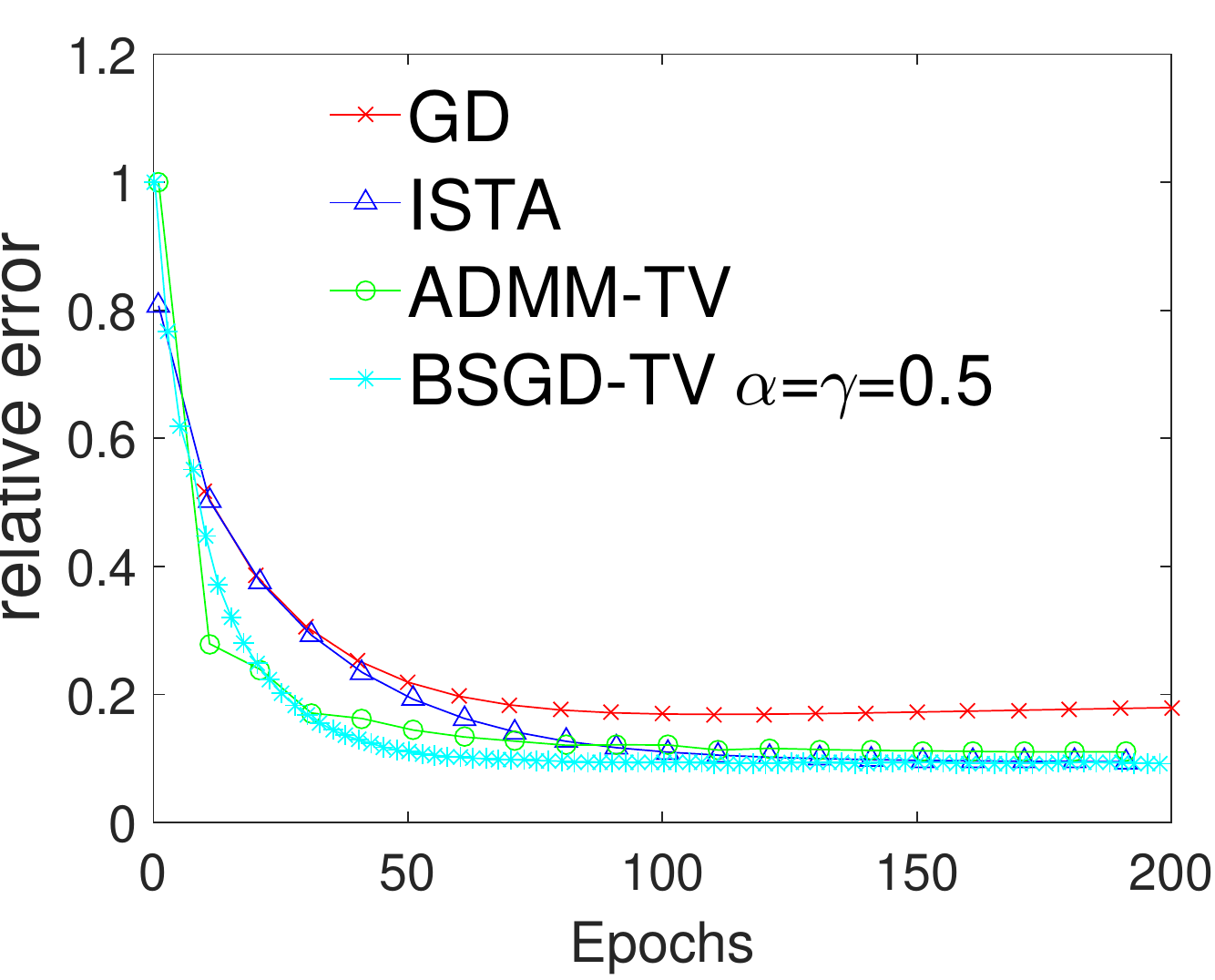} 
    
  } 
  \subfigure[]{ 
    \includegraphics[width=1.65in]{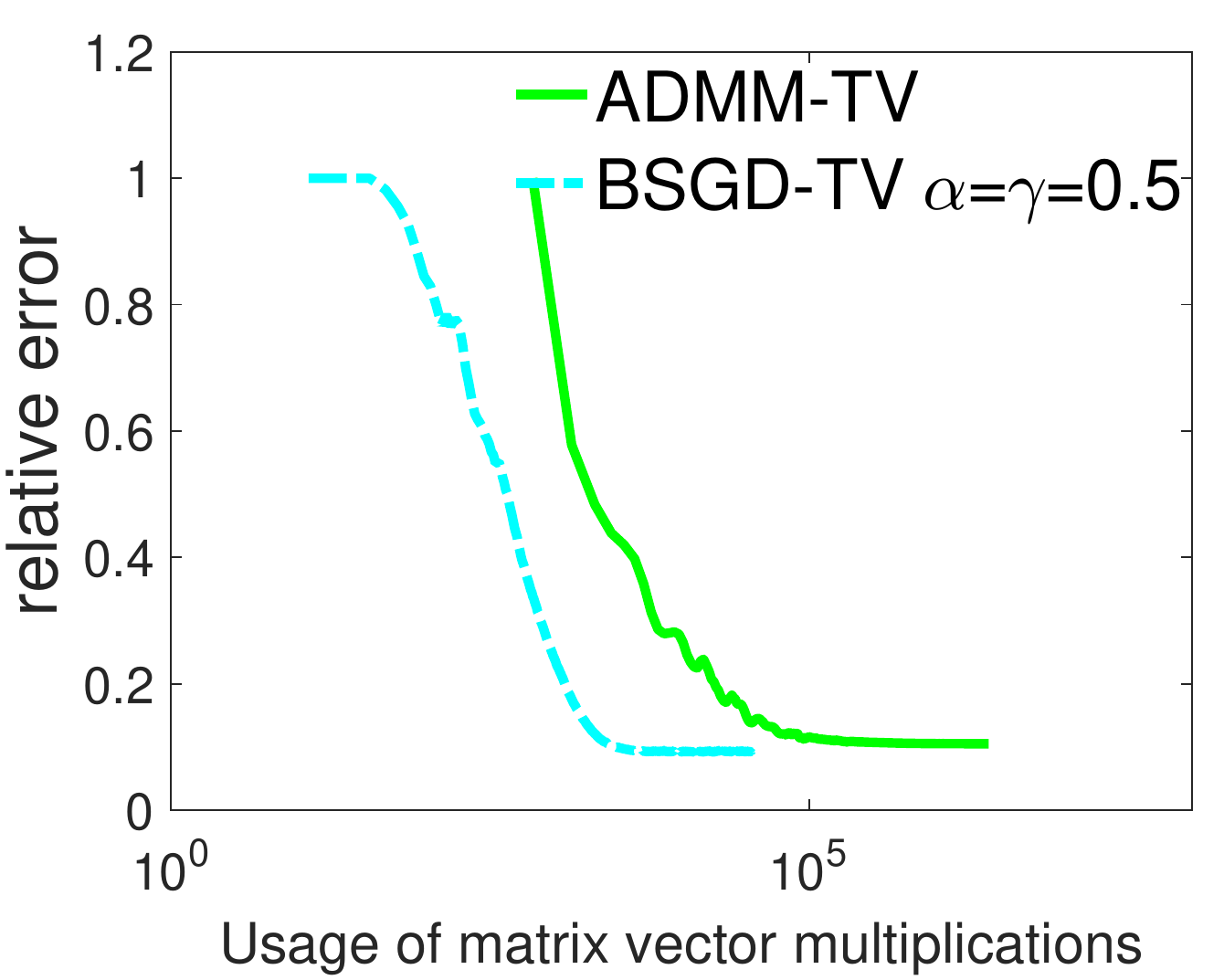} 
  } 
  \caption{The step length $\mu$ for ISTA, gradient descent (GD, which solves $f(\x)$ in Eq.\ref{Equ} ) and BSGD is $6e-4$ and $\lambda$ in Eq.\ref{Equ} is $0.1$. For BSGD and ADMM, $A$ is divided into $4\times 4$ sub-matrices while ISTA processes $\A$ as a whole. (a): Relative error vs. epoch. The high relative error of GD suggests the necessity of incorporating the TV norm.  (b):  BSGD-TV uses significantly fewer matrix-vector multiplications compared to ADMM-TV.}  
  \label{figsubfig} 
\end{figure}

%

\section{Conclusion}
\label{sec:second-section}

BSGD-TV is a parallel algorithm for large scale TV constrained CT reconstruction. It is similar to the popular ISTA algorithm but is specially designed for optimisation in distributed networks. The advantage is that individual compute nodes only operate on subsets of $\y$ and $\x$, which means they can operate with less internal memory. The method converges significantly faster than block-ADMM methods.   



\begin{thebibliography}{10}

\bibitem{guo2016convergence}
X.~Guo, ``Convergence studies on block iterative algorithms for image
  reconstruction'', \newblock Applied Mathematics and Computation, {\bf 273}: 525--534, 2016.
  
  \bibitem{sidky2008image}
E.~Sidky and X.Pan, ``Image reconstruction in circular cone-beam computed tomography by constrained, total-variation minimization'', \newblock Physics in Medicine $\&$ Biology, {\bf 53}(17):4777--4807, 2008.


\bibitem{gao2017distributed}
Y.~Gao and T.Blumensath, ``A Joint Row and Column Action Method for Cone-Beam Computed Tomography'', \newblock IEEE Transactions on Computational Imaging, 2018.

\bibitem{beck2009fast}
A.~Beck and M.Teboulle, ``A fast iterative shrinkage-thresholding algorithm for linear inverse problems'', \newblock SIAM journal on imaging sciences, {\bf 2}(1):183--202, 2009.

\bibitem{parikh2014block}
N.~Parikh and S.~Boyd, ``Block splitting for distributed optimization'', \newblock Mathematical Programming Computation, {\bf 6}(1):77--102, 2014.

\bibitem{beck2009fast2}
A.~Beck and M.Teboulle, ``Fast gradient-based algorithms for constrained total variation image denoising and deblurring problems'', \newblock IEEE Transactions on Image Processing, {\bf 18}(11):2419--2434, 2009.
\bibitem{saad2003iterative}
Y.~Saad, ``Iterative methods for sparse linear systems'', \newblock siam,{\bf 82}, 2003.


 












  
%

\end{thebibliography}
\end{document}